\documentclass[letterpaper, 10 pt, conference]{ieeeconf}  

\IEEEoverridecommandlockouts                              

\usepackage[pdftex, pdfstartview={FitV}, pdfpagelayout={TwoColumnLeft},bookmarksopen=true,plainpages = false, colorlinks=true, linkcolor=black, citecolor = black, urlcolor = black,filecolor=black , pagebackref=false,hypertexnames=false, plainpages=false, pdfpagelabels ]{hyperref}

\usepackage{graphicx}
\usepackage{epstopdf}
\usepackage{amssymb,amsmath}

\usepackage{amsthm}

\usepackage{subfig}
\usepackage[font=footnotesize]{caption}
\usepackage{xcolor}
\usepackage{units}
\usepackage{algorithm}
\usepackage[noend]{algpseudocode}
\usepackage{balance}
\usepackage[sort,compress]{cite}
\usepackage{tabularx}

\newtheoremstyle{named}{}{}{\itshape}{}{\bfseries}{.}{.5em}{\thmname{#1}\thmnumber{ #2}\thmnote{ (#3)}}
\theoremstyle{named}
\newtheorem{namedtheorem}{Theorem}
\theoremstyle{named}
\newtheorem{namedcorollary}{Corollary}

\newtheoremstyle{myproblem}{}{}{}{}{\bfseries}{--}{.5em}{\thmname{#1}\thmnumber{ #2} \thmnote{#3}}
\theoremstyle{myproblem}
\newtheorem{problem}{Problem}

\theoremstyle{plain}

\theoremstyle{definition}
\newtheorem{assumption}{Assumption}
\theoremstyle{definition}
\newtheorem{definition}{Definition}

\theoremstyle{remark}

\usepackage[hang,flushmargin]{footmisc}

\usepackage[capitalize]{cleveref}
\crefformat{equation}{(#2#1#3)}
\Crefformat{equation}{Equation~(#2#1#3)}
\Crefname{equation}{Equation}{Eqs.}

\usepackage{bm}

\usepackage{accents}
\newcommand{\ubar}[1]{\underaccent{\bar}{#1}}

\title{\LARGE \bf
Dynamic Tube MPC for Nonlinear Systems 
}

\author{Brett T. Lopez$^{1}$, Jean-Jacques E. Slotine$^{2}$, and Jonathan P. How$^{1}$
\thanks{$^{1}$Aerospace Control Laboratory, Massachusetts Institute of Technology, Cambridge MA, {\tt\small \{btlopez,jhow\}@mit.edu}}
\thanks{$^{2}$Nonlinear Systems Laboratory, Massachusetts Institute of Technology, Cambridge MA, {\tt\small jjs@mit.edu}}
}

\begin{document}

\maketitle
\thispagestyle{empty}
\pagestyle{empty}


\begin{abstract}

Modeling error or external disturbances can severely degrade the performance of Model Predictive Control (MPC) in real-world scenarios.
Robust MPC (RMPC) addresses this limitation by optimizing over feedback policies but at the expense of increased computational complexity.
Tube MPC is an approximate solution strategy in which a robust controller, designed offline, keeps the system in an invariant tube around a desired nominal trajectory, generated online.
Naturally, this decomposition is suboptimal, especially for systems with changing objectives or operating conditions.
In addition, many tube MPC approaches are unable to capture state-dependent uncertainty due to the complexity of calculating invariant tubes, resulting in overly-conservative approximations.
This work presents the Dynamic Tube MPC (DTMPC) framework for nonlinear systems where both the tube geometry and open-loop trajectory are optimized \textit{simultaneously}.
By using boundary layer sliding control, the tube geometry can be expressed as a simple relation between control parameters and uncertainty bound; enabling the tube geometry dynamics to be added to the nominal MPC optimization with minimal increase in computational complexity. 
In addition, DTMPC is able to leverage state-dependent uncertainty to reduce conservativeness and improve optimization feasibility.
DTMPC is demonstrated to robustly perform obstacle avoidance and modify the tube geometry in response to obstacle proximity.

\end{abstract}


\section{INTRODUCTION}
\label{sec:intro}

Model predictive control (MPC) has become a core control strategy because of its natural ability to handle constraints and balance competing objectives.
Heavy reliance on a model though makes MPC susceptible to modeling error and external disturbances, often leading to poor performance or instability.
Robust MPC (RMPC) addresses this limitation (at the expense of additional computational complexity) by optimizing over control policies instead of open-loop control actions.
Tube MPC is a tractable alternative that decomposes RMPC into an offline robust controller design and online open-loop MPC problem.
However, this decoupled design strategy restricts the tube geometry (i.e., feedback controller) to be fixed for \textit{all} operating conditions, which can lead to suboptimal performance. 
This article presents a framework for nonlinear systems where the tube geometry \textit{and} open-loop reference trajectory are designed simultaneously online, giving the optimization an additional degree of freedom to satisfy constraints or changing objectives.

Tube MPC for nonlinear systems has been an active area of research. 
For example, hierarchical MPC \cite{mayne2007tube}, reachability theory \cite{althoff2008reachability}, sliding mode control \cite{rubagotti2011robust,incremona2015hierarchical}, sum-of-square optimization \cite{majumdar2017funnel}, and Control Contraction Metrics \cite{singh2017robust} have all been recently used in nonlinear tube MPC.
These approaches try to maximize robustness by minimizing tube size given control constraints and bounds on uncertainty.
However, minimizing tube size typically results in a high-bandwidth controller that responds aggressively to measurement noise or external disturbances.
For mobile systems that use onboard sensing for estimation or perception, this type of response can severely degrade performance or cause a catastrophic failure. 
Further, the performance reduction often depends on the current operating environment so modifying the tube geometry online would be advantageous. 
While there is a precedent for optimizing tube geometry in linear MPC \cite{rakovic2012homothetic,rakovic2016elastic}, the relationship between tube geometry and control parameters for nonlinear systems is often too complex to put in a form suitable for real-time optimization. The approach described herein circumvents this issue by providing a simple and exact description of how the tube geometry, control parameters, and uncertainty are related, enabling the tube geometry to be optimized in real-time.

The primary contribution of this work is a tube MPC framework for nonlinear systems that simultaneously optimizes tube geometry and open-loop reference trajectories in the presence of uncertainty.
The proposed framework leverages the simplicity and strong robustness properties of time-varying boundary layer sliding control \cite{slotine1984sliding} to establish a connection between tube geometry, control parameters, and uncertainty.
Specifically, the tube geometry can be described by a simple first-order differential equation that is a function of control bandwidth and uncertainty bound.
This allows the development of a framework with several desirable properties.
First, the tube geometry can be easily optimized, with minimal increase in computational complexity, by treating the control bandwidth as a decision variable and augmenting the state vector with the tube geometry dynamics.
Second, the uncertainty bound in the tube dynamics can be made  state-dependent, allowing the optimizer to make smarter decisions about which states to avoid given the system's current state and proximity to constraints.
And third, less conservative tubes can be constructed by combining the tube and tracking error dynamics.
Simulation results demonstrate DTMPC's ability to optimize the tube geometry, via modulating control bandwidth and/or utilize knowledge of state-dependent uncertainty, to robustly avoid obstacles.   



\section{RELATED WORKS}
\label{sec:rw}

A number of works have been published on the stability, feasibility, and performance of linear tube MPC \cite{mayne2014model,rawlings2009model,mayne2016robust}.
While this is an effective strategy to achieve robustness, decoupling the nominal MPC problem and controller design is suboptimal. 
Rakovi\`c et al. \cite{rakovic2012homothetic} showed that the region of attraction can be enlarged by parameterizing the problem with the open-loop trajectory and tube size.
The authors presented the homothetic tube MPC (HTMPC) algorithm that treated the state and control tubes as homothetic copies of a fixed cross-section shape, enabling the problem to be parameterized by the tube’s centers (i.e., open-loop trajectory) and a cross-section scaling factor. 
The work was extended to tubes with varying shapes, known as elastic tube MPC (ETMPC), but at the expense of computational complexity \cite{rakovic2016elastic}. 
Both HTMPC and ETMPC possess strong theoretical properties and have the potential to significantly improve performance but a nonlinear extension has yet to be developed.

Recent theoretical and computational advances in nonlinear control design and invariant set computation has aided in the development of new nonlinear tube MPC techniques.
Mayne et al. proposed a two-tier MPC architecture where the nominal MPC problem, with tightened constraints, is solved followed by an ancillary problem that drives the current state to the nominal trajectory \cite{mayne2007tube}. 
Linear reachability theory is another strategy but tends to be overly conservative because nonlinearities are treated as disturbances \cite{althoff2008reachability}.
Because of its strong robustness properties, a number of works have proposed using sliding mode control as an ancillary controller \cite{muske2007predictive,rubagotti2011robust,mitic2013approach,chakrabarty2014robust,incremona2015hierarchical}. 
The work by Muske et al.\ is of particular interest because the parameters of the sliding surface were optimized within the MPC optimization to achieve minimum time state convergence.
Majumdar et al.\ constructed ancillary controllers for nonlinear systems via sum-of-squares (SOS) optimization that minimized funnel size (akin to a tube) \cite{majumdar2017funnel}. 
The method, however, required a pre-specified trajectory library and an extremely time consuming offline computation phase. 
Singh et al.\ proposed using Control Contraction Metrics to construct tubes and showed their approach increases the region of feasibility for the optimization \cite{singh2017robust}.
All of the aforementioned works fall into the category of rigid tube MPC (i.e., fixed tube size) so are inherently suboptimal.
Further, these approaches tend to produce overly conservative tubes because they cannot leverage knowledge of state-dependent uncertainty.

This work uses boundary layer sliding control to address the suboptimality and conservatism of the aforementioned techniques for nonlinear systems. This is accomplished by: 1) incorporating the tube geometry into the optimization, subsequently bridging the gap between linear and nonlinear homothetic/elastic tube MPC; 2) leveraging knowledge of state-dependent uncertainty; and 3) combining the tube and error dynamics to reduce the spread of possible trajectories.


\section{PROBLEM FORMULATION}
\label{sec:pf}
Consider a nonlinear, time-invariant, and control affine system given by (omitting the time argument)
\begin{equation}
\dot{x} = f\left(x\right) + b\left(x\right)u + d,
\label{eq:dyn}
\end{equation}
where $x \in \mathbb{R}^n$ is the state of the system, $u \in \mathbb{R}^m$ is the control input, and $d \in \mathbb{R}^n$ is an external disturbance. 

\begin{assumption}\label{assumption:dyn}
The dynamics $f$ can be expressed as \\$f=\hat{f}+\tilde{f}$ where $\hat{f}$ is the nominal dynamics and $\tilde{f}$ is the bounded model error (i.e., $|\tilde{f} (x)| \leq \Delta(x)$). 
\end{assumption}


Note that the model error bound in assumption \ref{assumption:dyn} is state-dependent, which can be leveraged to construct less conservative tubes.

\begin{assumption}
The disturbance $d$ belongs to  a closed, bounded, and connected set $\mathbb{D}$ (i.e., $\mathbb{D}:= \{ d \in \mathbb{R}^n : |d| \leq D \}$) and is in the span of the control input matrix (i.e., $d \in \text{span}\left(b(x)\right)$).
\end{assumption}

The standard RMPC formulation involves a minimax optimization to construct a feedback policy $\pi : \mathbb{X} ~\times~ \mathbb{R} \rightarrow \mathbb{U}$ where $x\in\mathbb{X}$ and $u\in\mathbb{U}$ are the allowable states and control inputs, respectively.
However, optimizing over arbitrary functions is not tractable and discretization suffers from the curse of dimensionality.
The standard approach taken in tube MPC \cite{mayne2014model} is to change the decision variable from control policy $\pi$ to open-loop control input $u^*$.
In order to achieve this re-parameterization, the following assumption is made about the structure of the control policy $\pi$.
\begin{assumption}\label{assumption:pi}
The control policy $\pi$ takes the form \\
$\pi = u^* + \kappa(x,x^*)$ where $u^*$ and $x^*$ are the open-loop input and reference trajectory, respectively.
\end{assumption}

In the tube MPC literature $\kappa$ is known as the \textit{ancillary controller} and is typically designed offline.
The role of the ancillary controller is to ensure the state $x$ remains in a \textit{robust control invariant} (RCI) tube around the nominal trajectory $x^*$. 

\begin{definition}\label{def:rci}
Let $\mathbb{X}$ denote the set of allowable states and let $\tilde{x} := x-x^*$. The set $\Omega \subset \mathbb{X}$ is a RCI tube if there exists an ancillary controller $\kappa\left(x,x^*\right)$ such that if $\tilde{x}\left(t_0\right) \in \Omega$, then, for all realizations of the disturbance and modeling error, $\tilde{x}(t) \in \Omega$, $\forall t \geq t_0$.
\end{definition}

\begin{figure}[t]
\centering
\includegraphics[width=0.4\textwidth]{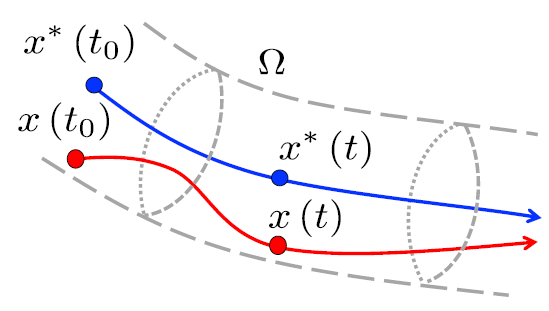}
\caption{Illustration of robust control invariant (RCI) tube $\Omega$ centered around desired state $x^*$. If the state $x$ begins in $\Omega$ then it remains in $\Omega$ indefinitely for all realizations of the model error or external disturbance.}
\label{fig:rci}
\vskip -0.2in
\end{figure}

\cref{fig:rci} provides a visualization of a RCI tube. Given an ancillary control $\kappa$ and associated RCI tube $\Omega$, a constraint-tightened version of the nominal MPC problem can be solved to generate an open-loop control input $u^*$ and trajectory $x^*$.

Calculating a RCI tube for a given ancillary controller can be difficult for nonlinear systems.
Unsurprisingly, the chosen methodology for synthesizing the ancillary controller can dramatically influence the complexity of calculating the tube geometry.
Ideally, the controller and tube geometry could be parameterized such that an explicit relationship between the two can be derived; enabling the controller and tube geometry to be designed online within the optimization.
Also, the control strategy should be able capture state-dependent uncertainty and how it impacts the tube geometry to reduce conservativeness.
While it may seem infeasible to find such a control synthesis strategy,  \cref{sec:blsc} will show that boundary layer sliding control possesses both properties.

\section{BOUNDARY LAYER SLIDING CONTROL}
\label{sec:blsc}

\subsection{Overview}
This section reviews time-varying boundary layer sliding control \cite{slotine1984sliding,slotine1991applied}, provides analysis supporting its use as an ancillary controller, and shows how the DTMPC framework leverages its properties. 
As reviewed in \cref{sec:rw}, sliding mode control has been extensively used for nonlinear tube MPC because of its simplicity and strong robustness properties.
Unlike other control strategies, sliding mode control completely cancels any bounded modeling error or external disturbance (reducing the RCI tube to zero). 
However, complete cancellation comes at the cost of high-frequency discontinuous control making it impractical for many real systems; a number of version that ensure continuity in the control signal have since been developed.
Note that the boundary layer controller was originally developed in \cite{slotine1984sliding} and is only presented here for completeness.
Before proceeding the following assumption is made.

\begin{assumption}
The system given by \cref{eq:dyn} has the same number of outputs to be controlled as inputs. More precisely, the dynamic can be expressed as
\begin{equation}
x_i^{(n_i)} = f_i(x)  + \sum_{j=1}^m b_{ij}(x) u_j + d_i, \quad i=1,...,m.
\end{equation}
\end{assumption}

Note that assumption~4 requires system \cref{eq:dyn} to be either feedback linearizable or minimum phase.
While many systems fall into one of these categories, future work will extend DTMPC to more general nonlinear systems.


\subsection{Sliding Control}

Let $\tilde{x}_i:=x_i-x^*_i$ be the tracking error for output $x_i$. Then, for $\lambda_i > 0$, the sliding variable $s_i$ for output $x_i$ is defined as 
\begin{align}
s_i &= \left(\frac{d}{dt} + \lambda_i\right)^{n_i-1} \tilde{x}_i \label{eq:s} \\
 &= \tilde{x}^{(n_i-1)}_i + \cdots + \lambda_i^{n_i-1} \tilde{x}_i \nonumber \\
 & = x_i^{(n_i-1)} - x_{ri}^{(n_i-1)}, \nonumber
\end{align}
\noindent where
\begin{equation}
x_{ri}^{(n_i-1)} = x_i^{*(n_i-1)} - \sum_{k=1}^{n_i-1}{{n_i-1}\choose{k-1}}\lambda_i^{n_i-k}\tilde{x}_i^{(k-1)}.
\end{equation}

In sliding mode control, a sliding manifold $\mathcal{S}_i$ is defined such that $s_i=0$ for all time once the manifold is reached.  
This condition guarantees the tracking error goes to zero exponentially via \cref{eq:s}.
It can be shown that a discontinuous controller is required to ensure the manifold $\mathcal{S}_i$ is reached in finite time and is invariant to uncertainty \cite{slotine1991applied}. 
However, high-frequency discontinuous control can, among other things, excite unmodeled high-frequency dynamics and shorten actuator life span.


One strategy to smooth the control input is to introduce a boundary layer around the switching surface.
Specifically, let the boundary layer be defined as $\mathcal{B}_i := \{x : |s_i| \leq \Phi_i\}$ where $\Phi_i$ is the boundary layer thickness.
If $\Phi_i$ is time varying, then the boundary layer can be made attractive if the following differential equation is satisfied
\begin{equation}
\frac{1}{2} \frac{d}{dt} s_i^2 \leq \left( \dot{\Phi}_i - \eta_i \right) |s_i|,
\label{eq:dist_s}
\end{equation}
where $\eta_i$ dictates the convergence rate to the sliding surface.
Differentiating \cref{eq:s},
\begin{align}
\dot{s}_i &= x_i^{(n_i)} - x_{ri}^{(n_i)} \nonumber \\
&= f_i(x)  + \sum_{j=1}^m b_{ij}(x) u_j + d_i - x_{ri}^{(n_i)}. \label{eq:sdot}
\end{align}
Stacking \cref{eq:sdot} for each output, the vector version is obtained
\begin{equation}
\dot{s} = F(x) + B(x)u + d + x_r^{(n)}. \label{eq:sdot_v}
\end{equation}
Note that $F$ and $B$ are stacked versions of the dynamics and input matrix, respectively, that correspond to the \textit{output variables}. 
If the output variables are chosen to be the full state vector (i.e., state feedback linearization), then 
$F$ and $B$ simply become the dynamics and input matrix in \cref{eq:dyn}.
  
Let the controller take the form
\begin{equation}
u = B(x)^{-1} \left[ -\hat{F}(x) -  x_r^{(n)} - K(x) \text{sat}\left( \nicefrac{s}{\Phi} \right) \right], \label{eq:u}
\end{equation}
where sat$(\cdot)$ is the saturation function and the division is element-wise.
Then, for $|s|>\Phi$, the boundary layer is attractive if
\begin{equation}
K(x) = \Delta(x)  + D + \eta - \dot{\Phi} .\label{eq:K}
\end{equation}

Addition information can be inferred by considering the sliding variable dynamics inside the boundary. 
Again substituting \cref{eq:u} into \cref{eq:sdot_v} with $|s| \leq \Phi$,
\begin{equation}
\dot{s} = -\frac{K(x)}{\Phi} s + F(x) - \hat{F}(x) + d, \label{eq:sdot_bl}
\end{equation}
where again the division is element-wise. Alternatively, \cref{eq:sdot_bl} can be written as
\begin{equation}
\dot{s} = -\frac{K(x^*)}{\Phi} s + \left(F(x^*) - \hat{F}(x^*) + d + O\left(\tilde{x}\right) \right), 
\label{eq:sdot_bl_d}
\end{equation}
which is a first order filter with cutoff frequency $\frac{K(x^*)}{\Phi}$.
Let $\alpha$ be the desired cutoff frequency, then, leveraging \cref{eq:K}, one obtains 
\begin{equation}
\frac{\Delta(x^*) + D + \eta - \dot{\Phi}}{\Phi} = \alpha ,
\end{equation}
or
\begin{equation}
\dot{\Phi} =  - \alpha \Phi + \Delta(x^*) + D + \eta. \label{eq:phidot}
\end{equation}
Thus, the final control law is given by \cref{eq:u}, \cref{eq:K}, and \cref{eq:phidot}.

\subsection{Discussion}
The boundary layer sliding controller in \cref{eq:u} allows us to establish several key properties at the core of DTMPC.


\begin{namedtheorem}[RCI Tube]
Let $\tilde{z}_i = \left[ \tilde{x}_i ~\dot{\tilde{x}}_i ~\cdots \right]^T$ be the error vector for output $\tilde{x}_i$.
Boundary layer control induces a robust control invariant tube $\Omega_i$ where the tube geometry is given by
\begin{equation}
\begin{aligned}
\Omega_i(t) \leq e^{A_{c,i}(t-t_0)} &  \Omega_i(t_0) ~ + \\ 
 &\int_{t_0}^t e^{A_{c,i}(t-t_0-\tau)}B_{c,i}\Phi_i(\tau) d\tau,
\end{aligned}
\label{eq:bound1}
\end{equation}
where $A_{c,i}$ and $B_{c,i}$ are found by putting \cref{eq:s} into the controllable canonical form. 
\label{thm:rci1}
\end{namedtheorem}

\begin{proof}
Recalling the definition of $s_i$ from \cref{eq:s}, the error dynamics are given by the linear differential equation
\begin{equation}
\tilde{x}_i^{(n_i-1)} + \cdots + \lambda_i^{n_i-1}\tilde{x}_i = s_i.
\label{eq:deqs}
\end{equation}
With the error vector $\tilde{z}_i = \left[ \tilde{x}_i ~\dot{\tilde{x}}_i ~\cdots \right]^T$ and  putting \cref{eq:deqs} into the controllable canonical form, the solution to \cref{eq:deqs} is
\begin{equation}
\tilde{z}_i (t) = e^{A_{c,i}(t-t_0)} \tilde{z}_i(t_0)+ \int_{t_0}^t e^{A_{c,i}(t-t_0-\tau)}B_{c,i} s_i d\tau.
\end{equation} 
Taking the element-wise absolute value $|\cdot|$, setting $\Omega_i(t) = |\tilde{z}_i(t)|$, and noting $|s_i| \leq \Phi_i$, \cref{eq:bound1} is obtained.
Thus, by \cref{def:rci}, $\Omega_i$ is a RCI tube since the error vector $\tilde{z}_i$ is bounded. 
\end{proof}

Theorem \ref{thm:rci1} proves that the geometry of the RCI tube $\Omega_i$ is uniquely described by the boundary layer thickness $\Phi_i$.
Using the terminology introduced by Rakovi\`c et al., the tubes in our approach are both homothetic and elastic.
For this reason, and the ability to capture state-dependent uncertainty, the approach developed here is called Dynamic Tube MPC.
Further, as briefly discussed in \cite{singh2017robust}, a tighter geometry can be obtained if the \textit{current} (as opposed to the predicted) tracking error is used in \cref{eq:bound1}.

The importance of \cref{eq:phidot} and \cref{eq:bound1} cannot be understated.
It gives a precise description for how the tube geometry changes with the level of uncertainty (from the model or otherwise). 
This is an incredibly useful relation for constructing tubes that are not overly conservative since, in most cases, the model error bound is typically picked to be a large constant because of the difficulty/inability to establish a relation like \cref{eq:phidot}. 
By letting the uncertainty be state-dependent, the controller and the MPC optimizer (to be discussed in \cref{sec:tgo}) can leverage all the available information to maximize performance.
This further underlines the importance of acquiring a high-fidelity model to reduce uncertainty and make the tube as small as possible without using high-bandwidth control.

Another interesting aspect of \cref{eq:phidot} is the choice of the cutoff frequency $\alpha$. 
In general, $\alpha$ and $\lambda$ are picked based on control-bandwidth requirements, such as  actuator limits or preventing excitation of higher-order dynamics.
It is clear from \cref{eq:phidot} that a larger $\alpha$ produces a smaller boundary layer thickness (i.e., high-bandwidth control leads to compact tubes).
However, from \cref{eq:sdot_bl_d}, increasing the bandwidth also increases the influence of the uncertainty.
Hence, the bandwidth should change depending on the current objective and proximity to state/control constraints (see \cref{sec:tgo}).


\section{DYNAMIC TUBE MPC}
\label{sec:tgo}

\subsection{Overview}
This section presents the DTMPC algorithm and discusses its properties.
DTMPC is a unique algorithm because of its ability to change the tube geometry to meet changing objectives and to leverage state-dependent uncertainty to maximize performance.
This section first presents a constraint tightening procedure necessary to prevent constraint violation due to uncertainty.
Next, optimizing the tube geometry by adding the control bandwidth as a decision variable is discussed.
Lastly, the non-convex formulation of DTMPC is presented.
Before proceeding, the following assumption is made about the form of the state and actuator constraints.
\begin{assumption}
The state and actuator constraints take the form 
\begin{equation}
\| P_x x + q_x \| \leq c_x, \quad \| P_u u + q_u \| \leq c_u,
\label{eq:constaints}
\end{equation} 
where $\|\cdot\|$ is the 2-norm.
\end{assumption}
\noindent Many physical systems posses these type of constrains so the above assumption is not overly restrictive. 

\subsection{Constraint Tightening}
State and actuator constraints must be modified to account for the nonzero tracking error and control input caused by model error and disturbances. 
The following corollary establishes the modified state constraint. 
\begin{namedcorollary}[Tightened State Constraint]
Assume the control law \cref{eq:u} is used as an ancillary controller with associated RCI tube $\mathcal{B}$ and bounded tracking error $|\tilde{x}|$. 
Then, the following modified state constraint
\begin{equation}
\| P_x x^* + q_x \| \leq c_x - \|P_x\tilde{x}\|,
\label{eq:xtight}
\end{equation}
guarantees, for all realization of the uncertainty, the true constraint is satisfied.
\label{cor:xbound}
\end{namedcorollary}
\begin{proof}
Recall that Theorem \ref{thm:rci1} established that the boundary layer controller induces is a RCI tube with geometry given by \cref{eq:bound1}. 
Then, the state is always upper bounded by $x \leq x^* + |\tilde{x}|$. Substituting this bound into the state constraint \cref{eq:constaints} and using the triangle inequality, the result is obtained. 
\end{proof}

Tightening the actuator constraints is more complicated since the control law in \cref{eq:u} depends on the current state $x$.
However, the tracking error bound can be used to obtain an upper bound on the control input that is only a function of the boundary layer thickness, desired state, and dynamics.
It is helpful to put the controller into a more useful form for the following theorem
\begin{equation}
\begin{split}
u = B(x)^{-1} \Biggl[  x^{*(n)}  -\hat{F}(x)-   \sum_{k=1}^{n-1} & {n-1 \choose k-1}  \lambda^{n-k}  \tilde{x}^{(k)}   \\
 & - K(x)\text{sat}\left(\nicefrac{s}{\Phi}\right) \Biggr],
\end{split}
\label{eq:uff}
\end{equation}
where the first term is the feedforward (and hence the decision variable in the optimization) and the last three are the feedback terms.

\begin{namedtheorem}[Control Input Upper Bound]
Assume that the control law is given by \cref{eq:uff}. Then, the control input is upper bounded, for all realizations of the uncertainty, by
\begin{equation}
u \leq \bar{B}^{-1}\Biggl[ x^{*(n)} +  \bar{F} + \sum_{k=1}^{n-1} {n-1 \choose k-1} \lambda^{n-k} | \tilde{x}^{(k)}|  + \bar{K} \Biggr],
\label{eq:ubound}
\end{equation}
where 
\begin{align}
\bar{B}^{-1} &= \max\left\{ B^{-1}\left(\ubar{x}\right),  B^{-1}\left(\bar{x}\right)\right\}, \label{eq:bbound} \\
\bar{F} &= \max\left\{ \left|\hat{F}\left(\ubar{x}\right)\right|,  \left|\hat{F}\left(\bar{x}\right)\right| \right\},\label{eq:fbound} \\
\bar{K} &= \max\left\{ K\left(\ubar{x}\right),  K\left(\bar{x}\right)\right\}, \label{eq:kbound}
\end{align}
with $\ubar{x}:= x^* - |\tilde{x}|$,  $\bar{x}:= x^* + |\tilde{x}|$, and $\max\left\{\cdot\right\}$ is the element-wise maximum.
\label{thm:ubound}
\end{namedtheorem}
\begin{proof}
The tracking error bound can be leveraged to eliminate the state-dependency in \cref{eq:uff}.
Specifically, the state is bounded by
\begin{equation}
x^* - |\tilde{x}| \leq x \leq x^* + |\tilde{x}|,
\label{eq:xbound}
\end{equation}
where $|\tilde{x}|$ is the solution to \cref{eq:bound1} when equality is imposed.
It is clear from \cref{eq:uff} that to upper bound $u$, the inverse of the input matrix $B^{-1}$ and the last three feedback terms should be maximized.
Define $\ubar{x}:=x^* - |\tilde{x}|$ and $\bar{x}:= x^* + |\tilde{x}|$, then using \cref{eq:xbound},  each term in \cref{eq:uff} can be upper bounded by evaluating at $\ubar{x}$ and $\bar{x}$ and taking the maximum, resulting in \cref{eq:bbound,eq:fbound,eq:kbound} and hence \cref{eq:ubound}.
\end{proof}

The bound established by Theorem \ref{thm:ubound} can be put into a more concise form
\begin{equation}
u \leq \bar{B}^{-1} \left[u^* + \bar{u}_{fb} \right],
\end{equation}
where $u^* := x^{*(n)}$ and $\bar{u}_{fb}$ is the sum of the last three terms in \cref{eq:ubound}.
Using Theorem \ref{thm:ubound}, the following corollary establishes the tightened actuator constraint.
\begin{namedcorollary}[Tightened Actuator Constraint]
Assume the control law \cref{eq:u} is used as an ancillary controller with associated RCI tube $\mathcal{B}$ and upper bound on input due to feedback $\bar{u}_{fb}$. 
Then, the following modified actuator constraint
\begin{equation}
\| P_u \bar{B}^{-1}u^* + q_u\| \leq c_u - \|P_u\bar{B}^{-1}\bar{u}_{fb}\|,
\label{eq:utight}
\end{equation}
guarantees, for all realization of the uncertainty, the true constraint is satisfied.
\label{cor:ubound}
\end{namedcorollary}
\begin{proof}
Theorem \ref{thm:ubound} established the upper bound on the control input to be $u \leq \bar{B}^{-1}\left[ u^* + \bar{u}_{fb}\right]$. Substituting this bound into the actuator constraint \cref{eq:constaints} and using the triangle inequality, the result is obtained. 
\end{proof}

\subsection{Optimized Tube Geometry}
For many autonomous systems, the ability to react to changing operating conditions is crucial for maximizing performance.
For instance, a UAV performing obstacle avoidance should modify the aggressiveness of the controller based on the current obstacle density to minimize expended energy.
Formally, the tube geometry must be added as a decision variable in the optimization to achieve this behavior.
DTMPC is able to optimize the tube geometry because of the simple relationship between the tube geometry, control bandwidth, and level of uncertainty given by \cref{eq:phidot}.
This is one of the distinguishing features of DTMPC since other state-of-the-art nonlinear tube MPC algorithms are not able to establish an explicit relationship like \cref{eq:phidot}.

In \cref{sec:blsc}, it was shown that the control bandwidth $\alpha$ is responsible for how the uncertainty affects the sliding variable $s$. 
Subsequently, the choice of $\alpha$ influences the tube geometry (via \cref{eq:phidot}) and control gain (via \cref{eq:K}).
In order to maintain continuity in the control signal, the tube geometry dynamics are augmented such that $\alpha$ and $\Phi$ remain smooth.
More precisely, the augmented tube dynamics are
\begin{align}
\dot{\Phi} &= -\alpha \Phi + \Delta(x^*) + D + \eta, \nonumber\\
\dot{\alpha} &= v,
\end{align}
where $v \in \mathbb{V}$ is an artificial input that will serve as an additional decision variable in the optimization. 
It is easy to show that the above set of differential equations is stable so long as $\alpha$ remains positive.

\subsection{Complete Formulation}
With Corollary \ref{cor:xbound} and \ref{cor:ubound} establishing the tightened state and actuator constraints, the Dynamic Tube MPC optimization can now be formulated as

\begin{problem}{Dynamic Tube MPC}
\begin{equation*}
\begin{aligned}
& \underset{\check{u}(t),\check{v}(t)}{\text{min}} & & J = h(\check{x}(t_f)) + \int\limits_{t_0}^{t_f} \ell (\check{x}(t),\check{u}(t),\check{\alpha}(t),\check{v}(t))dt \\
& \text{subject to} & & \dot{\check{x}}(t) = \hat{f}(\check{x}(t)) + b(\check{x}(t))\check{u}(t),  ~  \dot{\check{\alpha}}(t) = \check{v}(t), \\
& && \dot{\Phi}(t) = -\check{\alpha}(t) \Phi(t) + \Delta(\check{x}(t)) + D + \eta, \\
& && \dot{\Omega}(t) = A_c \Omega(t) + B_c \Phi(t), ~~ \Omega(t_0) = |\tilde{x}(t_0)|, \\
& && \check{x}(t_0) = x^*_0, ~~ \Phi(t_0) = \Phi_0, ~~ \check{x}\left(t_f\right) = x^*_f, \\
& && \check{x}(t) \in \mathbb{\bar{X}}, ~~ \check{u}(t) \in \mathbb{\bar{U}}, ~~ \check{\alpha}(t) \in \mathbb{A} , ~~ \check{v}(t) \in \mathbb{V},
\end{aligned}
\end{equation*}
\label{prob:dtmpc}
\end{problem}
\vskip -0.1in

\noindent where $\check{\cdot}$ denotes the internal variables in the optimization; $\Omega$ is the tube geometry with matrices $A_c$ and $B_c$ given by putting \cref{eq:s} into controllable canonical form; $\mathbb{\bar{X}}$ and  $\mathbb{\bar{U}}$ are the tightened state and actuator constraints; and $\ell$ and $h$ are the quadratic state and terminal cost.
The output of DTMPC is an optimal open-loop (i.e., feedforward) control input $u^*$, trajectory $x^*$, and control bandwidth $\alpha^*$. 

DTMPC is inherently a non-convex optimization problem because of the nonlinear dynamics.
However, non-convexity is a fundamental characteristic of nonlinear tube MPC and a number of approximate solution procedures have been proposed.
The key takeaway, though, is that Problem \ref{prob:dtmpc} is a nonlinear tube MPC algorithm that \textit{simultaneously} optimizes the open-loop trajectory and tube geometry, eliminating the duality gap in standard tube MPC.
Furthermore, conservativeness can be reduced since Problem \ref{prob:dtmpc} is able to leverage state-dependent uncertainty to select an open-loop trajectory based on the structure of the uncertainty and proximity to constraints.
The benefits of these properties, in addition to combining the tube geometry and error dynamics, will be demonstrated in \cref{sec:results}.



\section{COLLISION AVOIDANCE MODEL}
\label{sec:cam}

\subsection{Overview}
Collision avoidance is a fundamental capability for many autonomous systems, and is an ideal domain to test DTMPC for two reasons.
First, enough safety margin must be allocated to prevent collisions when model error or disturbances are present.
More precisely, the optimizer must leverage knowledge of the peak tracking error (given by the tube geometry) to prevent collisions. 
The robustness of DTMPC and ability to utilize knowledge of state dependent uncertainty can thus be demonstrated.
Second, many real-world operating environments have variable obstacle densities so the tube geometry can be optimized in response to a changing environment.
The rest of this section presents the model and formal optimal control problem.

\subsection{Model}
This work uses a double integrator model with nonlinear drag, which describes the dynamics of many mechanical systems.
Let $r=\left[r_x~r_y~r_z\right]^T$ be the inertial position of the system that is to be tracked.
The dynamics are
\begin{equation}
\ddot{r} = - C_d \left\|  \dot{r} \right\| \dot{r} + g + u + d, 
\label{eq:model}
\end{equation}
where $g \in \mathbb{R}^3$ is the gravity vector, $C_d$ is the unknown but bounded drag coefficient ($0 \leq C_d \leq \bar{C}_d$), and $d$ is a bounded disturbance ($|d| \leq D$). 
From \cref{eq:u}, the control law is  
\begin{equation}
u =  \hat{C}_d \left\|  \dot{r} \right\|  \dot{r} + \ddot{r}^* - \lambda \dot{\tilde{r}} - K \text{sat}\left(\nicefrac{s}{\Phi}\right), 
\end{equation}
where $\hat{C}_d$ is the best estimate of the drag coefficient, $s = \dot{\tilde{r}} + \lambda \tilde{r}$, and 
\begin{align}
K &= \bar{C}_d \left\| \dot{r} \right\| \left| \dot{r} \right|  - \bar{C}_d \left\| \dot{r}^* \right\| \left| \dot{r}^* \right| + \alpha \Phi, \\
\dot{\Phi} &= - \alpha^* \Phi + \bar{C}_d \left\|  \dot{r}^* \right\|  \left| \dot{r}^* \right| + D + \eta.
\label{eq:phi_model}
\end{align}

\subsection{Collision Avoidance DTMPC}
Let $H$, $p_{c}$, and $r_{o}$ denote the shape, location, and size of an obstacle.
The minimum control effort DTMPC optimization with collision avoidance for system \cref{eq:model} is formulated as
\begin{problem}{Collision Avoidance DTMPC}
\begin{equation*}
\begin{aligned}
& \underset{\check{u}(t),\check{v}(t)}{\text{min}} & & J =  \int\limits_{t_0}^{t_f} \left[ \check{u}(t)^TQ\check{u}(t)+ \tilde{\alpha}(t)^T R	 \tilde{\alpha}(t) \right]dt \\
& \text{subject to} & & \ddot{\check{r}}(t) = - \hat{C}_d \left\|  \dot{\check{r}}(t) \right\| \dot{\check{r}}(t) + g + \check{u}(t),  ~  \dot{\alpha}(t) = v(t), \\
& && \dot{\check{\Phi}}(t) = - \check{\alpha}(t) \check{\Phi}(t) + \bar{C}_d \left\|  \dot{\check{r}}(t) \right\|  \left| \dot{\check{r}}(t) \right| + D + \eta, \\
& && \dot{\check{\Omega}}(t) = A_c\check{\Omega}(t) + B_c \check{\Phi}(t), ~~ \check{\Omega}(t_0) = |\tilde{r}(t_0)|, \\
& && \check{r}\left(t_0\right) = r^*_0, ~~ \check{\Phi}\left(t_0\right) = \Phi_0, ~~ \check{r}\left(t_f\right) = r^*_f, \\
& && \|H_i r(t) - p_{c,i} \| \geq r_{o,i} +  \| H_i\tilde{r}(t)\|, ~ i=1:N_{o}, \\
& && | \dot{\check{r}}(t) | \leq \dot{r}_m - |\dot{\tilde{r}}|, ~~  \| u^*(t) \| \leq u_m - \bar{u}_{fb}, \\
& &&  |v(t)| \leq v_m,~0 < \ubar{\alpha} \leq \check{\alpha}(t) \leq \bar{\alpha}, ~ \tilde{\alpha}=\check{\alpha}(t) - \ubar{\alpha}, 
\end{aligned}
\end{equation*}
\label{prob:cadtmpc}
\end{problem}
\vskip -0.15in
\noindent where again $\check{\cdot}$ denotes the internal variables of the optimization, $|\cdot|$ is the element-wise absolute value, $\ubar{\alpha}$ and $\bar{\alpha}$ are the upper and lower bounds of the control bandwidth, $\dot{r}_m$ is the peak desired speed, $v_m$ is the max artificial input, and $N_o$ is the number of obstacles.


\section{SIMULATION ENVIRONMENT}
\label{sec:sim}

DTMPC was tested in simulation to demonstrate its ability to optimize tube geometry and utilize knowledge of state-dependent uncertainty through an environment with obstacles. 
The obstacles were placed non-uniformly to emulate a changing operating condition (i.e., dense/open environment).
In order to emphasize both characteristics of DTMPC, three test cases were conducted.
First, the bandwidth was optimized when both the model and obstacle locations were completely known.
Second, the bandwidth was again optimized with a known model but the obstacle locations were unknown, requiring a receding horizon implementation.
Third, state-dependent uncertainty is considered but control bandwidth is kept constant.
Nothing about the formulation prevents optimizing bandwidth and leveraging state-dependent uncertainty simultaneously in a receding horizon fashion, this decoupling is only for clarity.
The tracking error \cref{eq:bound1} is used to tighten the obstacle and velocity constraint.

Problem \ref{prob:cadtmpc} is non-convex due to the nonlinear dynamics and non-convex obstacle constraints so sequential convex programming, similar to that in \cite{mao2016successive}, was used to obtain a solution. 
The optimization was initialized with a na\"ive straight-line solution and solved using YALMIP \cite{Lofberg2004} and MOSEK \cite{mosek2010mosek} in MATLAB.
If large perturbations to the initial guess are required to find a feasible solution, then warm starting the optimization with a better initial guess (possibly provided by a global geometric planner) might be necessary.
For the cases tested in this work, the optimization converged within three to four iterations -- fast enough for real-time applications.
The simulation parameters are summarized in \cref{tab:params}.

\begin{table}[t!]
	\begin{center}
		\caption{Simulation Parameters.}
		\label{tab:params}
		\begin{tabularx}{0.45\textwidth}{X X | X  X}
			\hline \hline 
			\textbf{Param.} & \textbf{Value} & \textbf{Param.} & \textbf{Value}  \\
			\hline \rule{0pt}{1.1\normalbaselineskip}$r_0$ & [0 0 1]$^T$ m& $r_f$ & [0 25 1]$^T$ m \\ 
			$\dot{r}_0$ & [0 1 0]$^T$ m/s& $\dot{r}_f$ & [0 1 0]$^T $m/s \\
			$\lambda$ & [2 2 2]$^T$ rad/s & $R_f$ & 2$I_3$ \\ 
			$Q$ & 2$I_3$ & $R$ & 0.1$I_3$ \\ 
			$\ubar{\alpha}$ & 0.5 rad/s & $\bar{\alpha}$ & 4 rad/s \\
			$u_m$ & 5 m/s$^2$ & $v_m$ & 2 rad/s$^2$  \\
			$\dot{r}_m$ & 2.5 m/s & $D$ & 0.5 m/s$^2$   \\
			$\hat{C}_d$ & 0.1 kg/m & $\bar{C}_d$ & 0.2 kg/m\\
			$t_f$ & 14 s & $N_o$ & 5 \\
			$\eta$ & 0.1 rad/s$^2$ & - & -\\
			\hline \hline
		\end{tabularx}
	\end{center}
	\vskip -0.3in
\end{table}


\section{RESULTS AND ANALYSIS}
\label{sec:results}

\subsection{Optimized Tube Geometry}

The first test scenario for DTMPC highlights its ability to simultaneously optimize an open-loop trajectory and tube geometry in a known environment with obstacles placed non-uniformly.
\cref{fig:dtmpc_bw} shows the open-loop trajectory (multi-color), tube geometry (black), and obstacles (grey) when DTMPC optimizes both the trajectory and tube geometry.
The color of the trajectory indicates the spatial variation of the control bandwidth, where low- and high-bandwidth are mapped to dark blue and yellow, respectively. 
It is clear that the bandwidth changes dramatically along the trajectory, especially in the vicinity of obstacles.
The insets in \cref{fig:dtmpc_bw} show that high-bandwidth (compact tube geometry) is used for the narrow gap and slalom and low-bandwidth (large tube geometry) for open space.
Hence, high-bandwidth control is only used when the system is in close proximity to constraints (i.e., obstacles), consequently limiting aggressive control inputs to only when they are absolutely necessary.
Thus, DTMPC can react to varying operating conditions by modifying the trajectory and tube geometry appropriately.

Since the tube geometry changes dramatically along the trajectory, it is important to verify that the tube remains invariant.
This was tested by conducting 1000 simulations of the closed-loop system with a disturbance profile sampled uniformly from the disturbance set $\mathbb{D}$.
\cref{fig:dtmpc_bw_mc} shows the nominal trajectory (red), each closed-loop trial run (blue), tube geometry (black), and obstacles (grey).
The inserts show that the state stays within the tube, even as the geometry changes, which verifies that the time-varying tube remains invariant. 

\subsection{Receding Horizon Optimized Tube Geometry}

In many situations the operating environment is not completely known and requires a receding horizon implementation.
The second test scenario for DTMPC highlights its ability to simultaneously optimize an open-loop trajectory and tube geometry in a unknown environment.
\cref{fig:rhc} shows a receding horizon implementation of DTMPC where only a subset of obstacles are known (dark-grey) and the rest are unknown (light-grey).
The bandwidth along the trajectory is visualized with the color map where low- and high-bandwidth are mapped to dark blue and yellow.
The first planned trajectory (\cref{fig:rhc-1}) uses high-bandwidth at the narrow gap and low-bandwidth in open space. 
When the second and third set of obstacles are observed, \cref{fig:rhc-2} and \cref{fig:rhc-3} respectively, DTMPC modifies the trajectory to again use high-bandwidth when in close-proximity to newly discovered obstacles.
This further demonstrates DTMPC's ability to construct an optimized trajectory and tube geometry in response to new obstacles.

\begin{figure}[t!]
\centering
\includegraphics[width=0.452\textwidth]{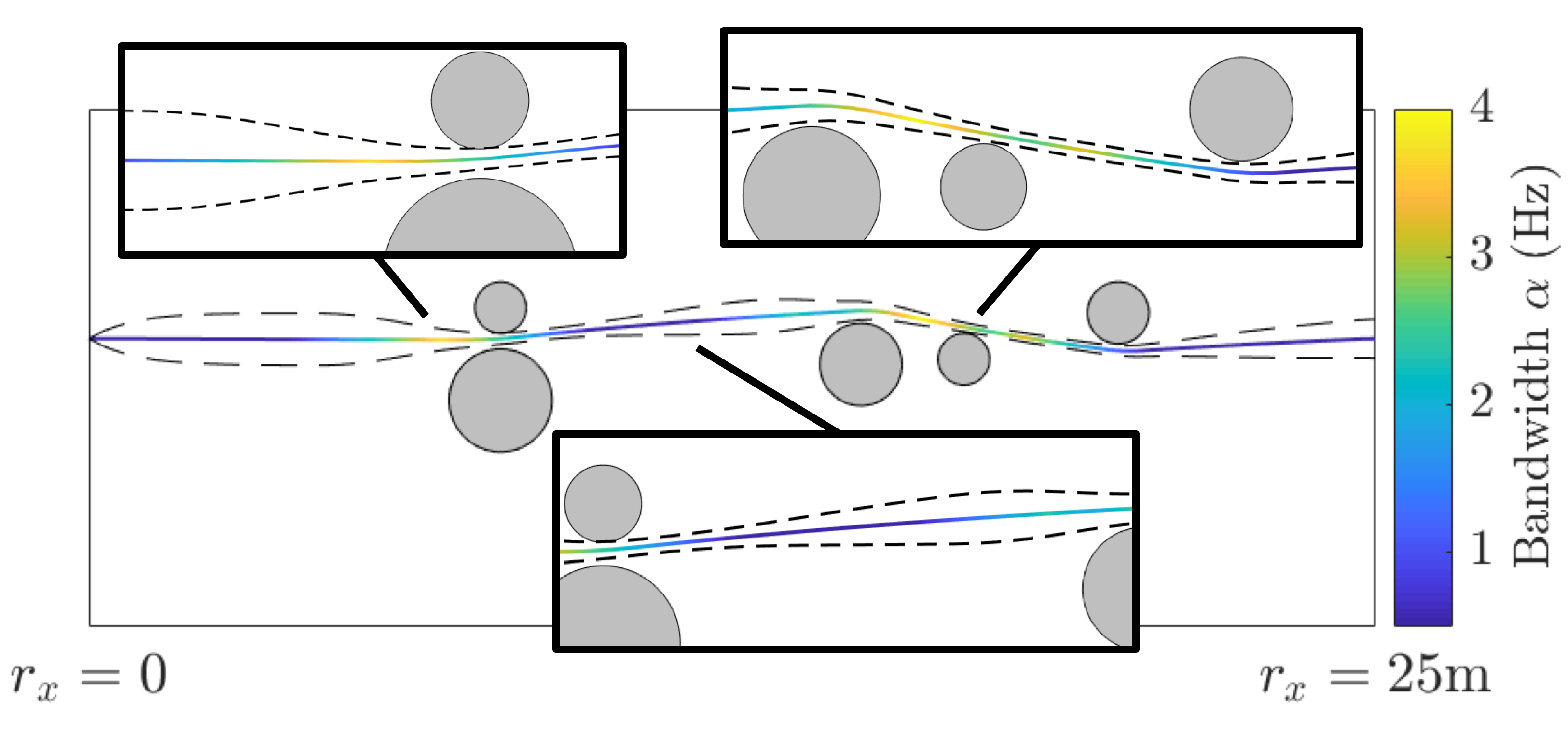}
 \vskip -0.01in
\caption{DTMPC simultaneously optimizing an open-loop trajectory (multi-color) and tube geometry (black) around obstacles (grey). High-bandwidth control (yellow) is used when in close proximity to obstacles while low-bandwidth control (dark blue) is used in open space.}
\label{fig:dtmpc_bw}
%
\centering
\includegraphics[width=0.452\textwidth]{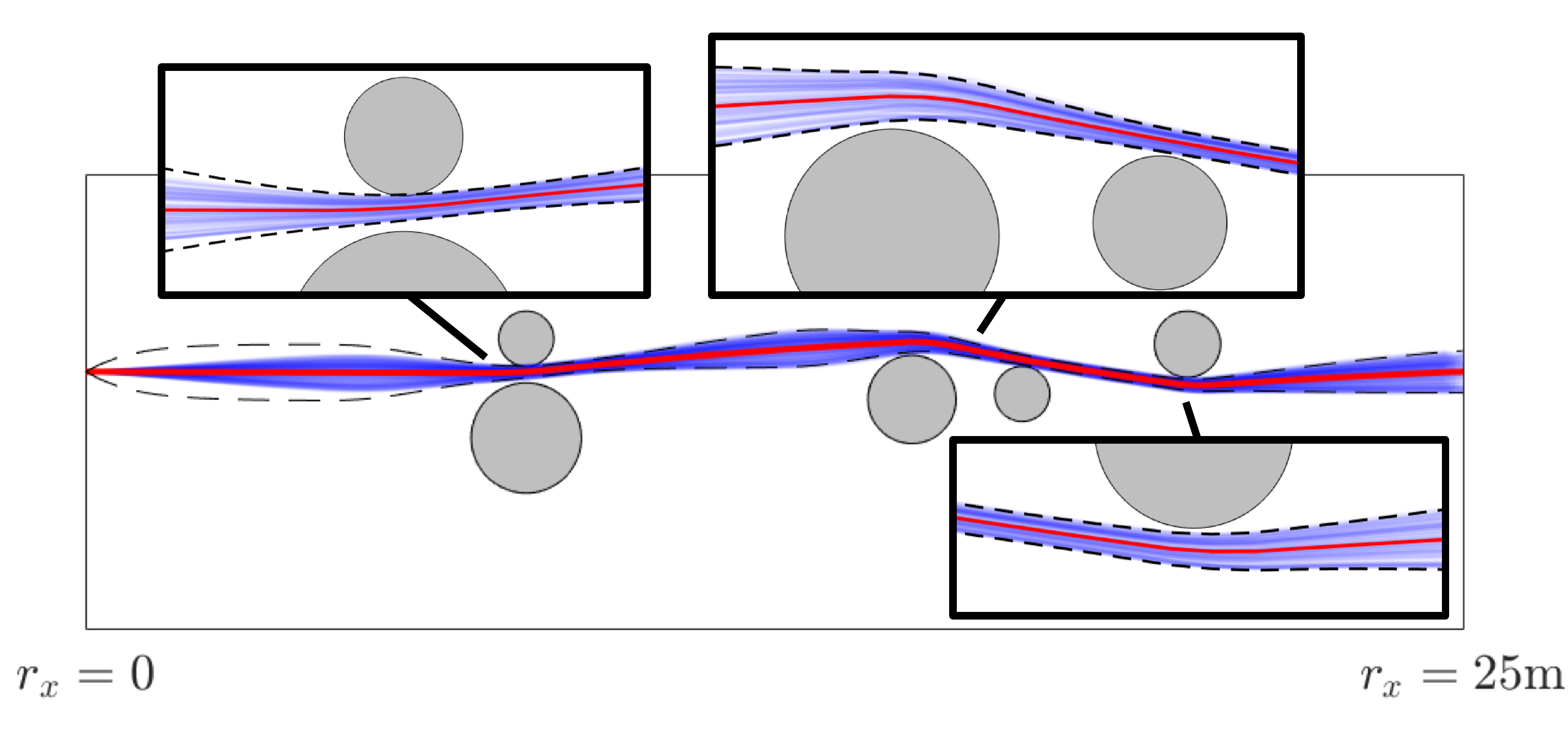}
 \vskip -0.01in
\caption{Monte Carlo verification that the time-varying boundary layer in DTMPC remains a robust control invariant tube. The closed-loop system (blue) was simulated with a different disturbance profile uniformly sampled from the disturbance set.}
\label{fig:dtmpc_bw_mc}
\vskip -0.2in
\end{figure}

\begin{figure}[!t]
\centering
\subfloat[Planned trajectory at $t$ = 0s.]{\includegraphics[width=0.452\textwidth]{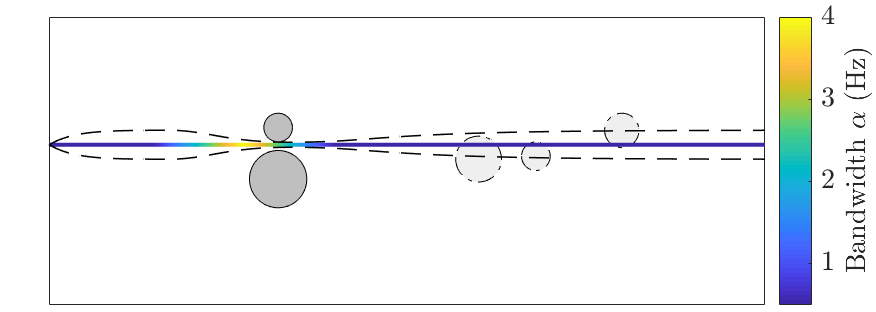}
\label{fig:rhc-1}}
 \\
 \vskip -0.01in
\subfloat[Planned trajectory a $t$ = 6s.]{\includegraphics[width=0.452\textwidth]{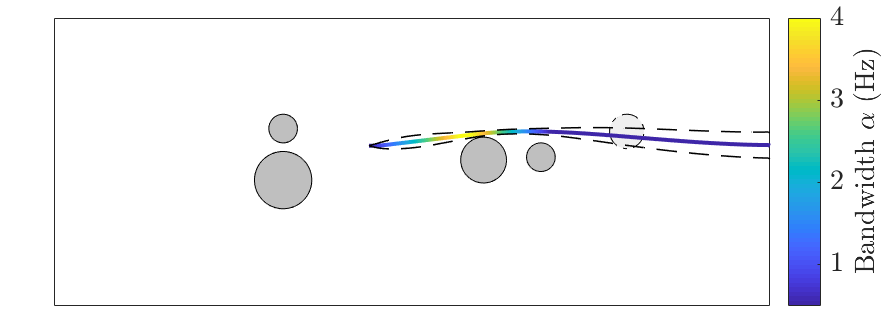}
\label{fig:rhc-2}}
\\
 \vskip -0.01in
\subfloat[Planned trajectory at $t$ = 8s.]{\includegraphics[width=0.452\textwidth]{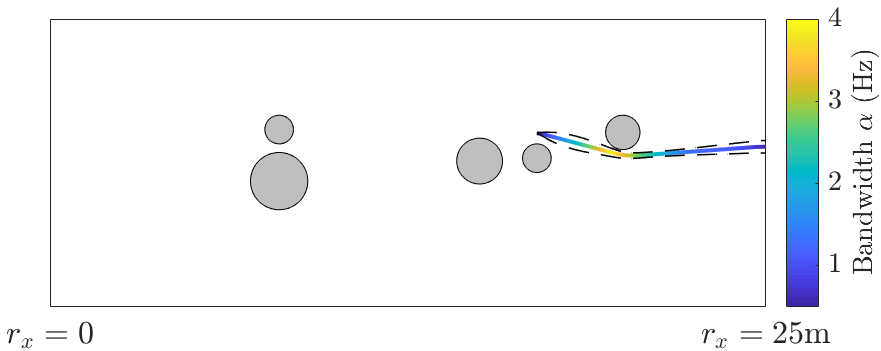}
\label{fig:rhc-3}}
 \vskip -0.01in
\caption{Receding horizon implementation of DTMPC with known (dark-grey) and unknown (light-grey) obstacles. The bandwidth along trajectory (multi-color) varies, resulting in a dynamic tube geometry (black). (a): First planned trajectory and tube geometry when only the first two obstacles are known. (b): New planned trajectory when the next two obstacles are observed. (c): New planned trajectory when the last obstacle is observed. }
\label{fig:rhc}
\vskip -0.2in
\end{figure}
 
\subsection{State-Dependent Uncertainty}

The third test scenario for DTMPC highlights its ability to leverage knowledge of state-dependent uncertainty, in this case arising from an unknown drag coefficient.
From \cref{eq:phi_model}, the uncertainty scales with the square of the velocity so higher speeds increase uncertainty.
\cref{fig:dtmpc_state} shows the open-loop trajectory (multi-color), tube geometry (black), and obstacles (grey) when DTMPC leverages state-dependent uncertainty.
The color of the trajectory is an indication of the instantaneous speed, where low and high speed are mapped to black and peach, respectively.  
It is clear that DTMPC generates a speed profile modulated by proximity to obstacles.
For instance, using the insets in \cref{fig:dtmpc_state}, the speed is lower (darker) when the trajectory goes through the narrow gap and around the other obstacles; reducing uncertainty and tightening the tube geometry.
Further, the speed is higher (lighter) when in the open, subsequently increasing uncertainty causing the tube geometry to expand.
If the state-dependent uncertainty is just assumed to be bounded, a simplification often made out of necessity in other tube MPC algorithms, the tube geometry is so large that, for this obstacle field, the optimization is infeasible with the same straight-line initialization as DTMPC.
Hence, DTMPC is able to leverage knowledge of state-dependent uncertainty to reduce conservatism and improve feasibility.

\begin{figure}[t!]
\centering
\includegraphics[width=0.452\textwidth]{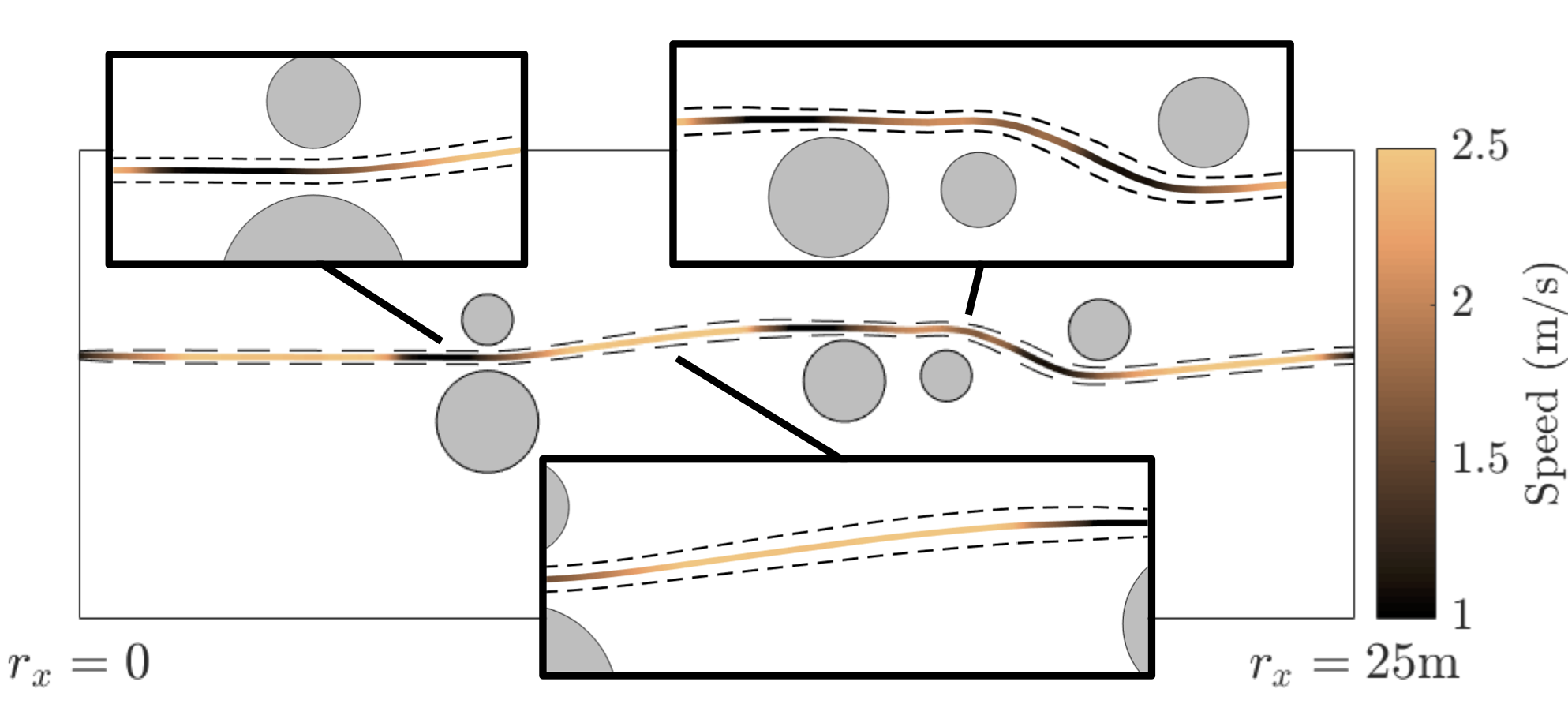}
 \vskip -0.01in
\caption{DTMPC leveraging state-dependent uncertainty to robustly avoid obstacles (grey). The speed along the trajectory, given by the color map, is low (dark) when in close proximity to obstacles and is high (light) in open regions. This causes the tube geometry (black) to contract and expand.}
\label{fig:dtmpc_state}
\vskip -0.2in
\end{figure}



\section{CONCLUSIONS}
This work presented the Dynamic Tube MPC (DTMPC) algorithm that addresses a number of shortcomings of existing nonlinear tube MPC algorithms. 
First, the open-loop MPC optimization is augmented with the tube geometry dynamics enabling the trajectory and tube to be optimized simultaneously. 
Second, DTMPC is able to utilize state-dependent uncertainty to reduce conservativeness and improve optimization feasibility.
And third, the tube geometry and error dynamics can be combined to further reduce conservativeness.
All three of these properties were made possible by leveraging the simplicity and robustness of boundary layer sliding control.
Simulation results showed that DTMPC is able to control the tube geometry size, by changing control bandwidth or leveraging state-dependent uncertainty, in response to changing operating conditions.
Future work includes expanding DTMPC to more general nonlinear systems. 

\addtolength{\textheight}{-12cm}   



\section*{ACKNOWLEDGMENT} 
This material is based upon work supported by the National Science Foundation Graduate Research Fellowship under Grant No. 1122374, by the DARPA Fast Lightweight Autonomy (FLA) program, by the NASA Convergent Aeronautics Solutions project Design Environment for Novel
Vertical Lift
Vehicles (DELIVER), and by ARL DCIST under Cooperative Agreement Number W911NF-17-2-0181.

\balance
\bibliographystyle{ieeetr}
\newcommand{\noopsort}[1]{} \newcommand{\printfirst}[2]{#1}
  \newcommand{\singleletter}[1]{#1} \newcommand{\switchargs}[2]{#2#1}


\end{document}